\documentclass{article}

\usepackage{amsthm} 
\usepackage{amssymb} 
\usepackage{graphicx} 
\usepackage{hyperref} 

\usepackage{lineno}

\newtheorem{theorem}{Theorem}

\theoremstyle{definition}

\newcommand\unnumberedfootnote[1]{%
  \begingroup
  \renewcommand\thefootnote{}%
  \footnotetext{#1}%
  \endgroup
}

\begin{document}

\title{Optimally Covering Large Triangles with Homothetic Unit Triangles}

\unnumberedfootnote{MSC: 52C15}

\author{John M. Boyer \footnote{Email: jboyer@acm.org; ORCID: 0000-0002-4755-5535}}
\maketitle

\begin{abstract}
We answer an open problem in the \emph{American Mathematical Monthly} about covering large triangles. Given a triangle $T$ of any triangular shape with a selected side length between $n \in \mathbb{N}$ and $n+1$, Baek and Lee proved that $T$ could not be covered with $n^2+1$ homothetic unit triangles (with the selected side of length 1). Letting $T_{n+d}$ denote a triangle with selected side length $n + d$ with $d \in (0, 1)$, Baek and Lee extended their proof to establish upper bounds for $d$ above which a $T_{n+d}$ cannot be covered with $n^2+2$ or $n^2+3$ homothetic unit triangles. Then, they showed that these bounds are tight based on analyses of a method by Conway and Soifer for the $n^2+2$ case and their own method for the $n^2+3$ case. Baek and Lee stated as an open problem the need to find tight upper bounds for the $n^2 + k$ cases for $4 \le k \le 2n$. We extend the Baek and Lee proof to establish upper bounds for those higher cases, and we show the upper bounds are tight by presenting two new triangle covering methods for the odd and even cases of $k$ that meet the upper bounds, as well as an optimal consolidated method that uses whichever of the two will cover a given $T_{n+d}$ with the fewest homothetic unit triangles.
\end{abstract}


\section{Introduction}\label{sec:Intro}

Let $T_s$ denote a triangle with selected side length $s \ge 1$. Given a triangle $T_1$ with the selected side of unit length, i.e., a \textbf{\emph{unit triangle}}, it is well-known that a larger triangle $T_n$ with selected side length $n \in \mathbb{N}$ can be covered with $n^2$ homothetic unit triangles. An example covering for an equilateral triangle $T_3$ is given in Figure~\ref{fig:CoverTriangleTn}(a). The unit triangles with a horizontal leg at the base and an acme above it are called \textbf{\emph{up-triangles}}, such as the one labeled `U' in Figure~\ref{fig:CoverTriangleTn}(a). Similarly, unit triangles with a horizontal leg at the top and a nadir below it are called \textbf{\emph{down-triangles}}, such as the one labeled `D' in Figure~\ref{fig:CoverTriangleTn}(a). The covering has $n$ rows, numbered 1 to 3 in Figure~\ref{fig:CoverTriangleTn}(a). Each row $k$ contains $k$ up-triangles, and each up-triangle except the last is succeeded in the row by one of $k-1$ down-triangles. Each row $k$ has the shape of a trapezoid with a top width of $k - 1$, diagonals of length 1, and a base width of $k$ units (the top row is triangular, which is degenerately trapezoidal).

The sense of direction is imparted only to aid exposition. Similarly, we focus on expressing results for equilateral triangles, which also presents no loss of generality because of the requirement of using homothetic unit triangles. As Soifer expressed it, ``for \emph{every} non-equilateral triangle $T$ \ldots an appropriate affine transformation maps the equilateral triangle and its covering \ldots into $T$"~\cite[p. 34]{Soifer}. 

Given a triangle $T_{n+d}$ with $n \in \mathbb{N}$ and $d \in (0, 1)$, a covering with only $n^2$ homothetic unit triangles leaves a coverage gap, such as the long, thin, darkened trapezoid in Figure~\ref{fig:CoverTriangleTn}(b). One way to deal with the coverage gap is to use the na\"{\i}ve covering method, which covers a $T_{n+d}$ with a covering of $T_{n+1}$. The na\"{\i}ve covering method requires $n^2+2n+1$ homothetic unit triangles. However, when $d$ is a small positive value $\varepsilon$ (which is vanishingly small as $n$ increases), then Conway and Soifer showed two methods to cover $T_{n+\varepsilon}$ using only $n^2+2$ homothetic unit triangles~\cite{ConwaySoiferGeom, ConwaySoifer}.

\begin{figure}[!ht]
\centering
\includegraphics[width=4.7in,height=2.4in]{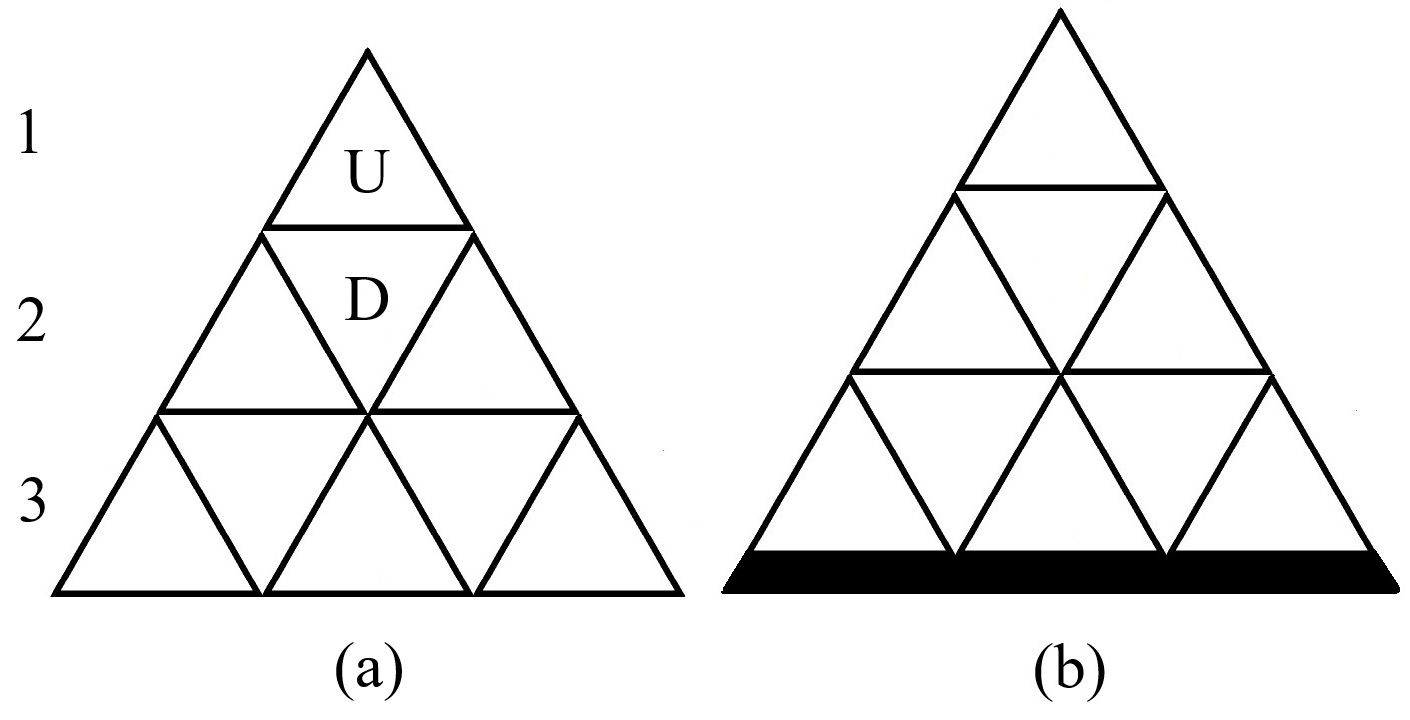}
\caption{(a) The well-known covering of a large triangle $T_n$ with $n \in \mathbb{N}$ (for $n=3$). (b) Depiction of the trapezoidal gap in coverage of a $T_{n+d}$ with only $n^2$ homothetic unit triangles.
\label{fig:CoverTriangleTn}}
\end{figure}

According to~\cite{Soifer}, Soifer's originating question was whether the darkened trapezoid in Figure~\ref{fig:CoverTriangleTn}(b) could, if its area were low enough, be covered by rearranging the unit triangles of a $T_{n}$ covering and then adding 1 extra unit triangle. The solutions that Conway and Soifer presented in~\cite{ConwaySoiferGeom, ConwaySoifer} added $2$ extra unit triangles, and Baek and Lee~\cite{BaekLee} recently proved that it is not possible to use only $n^2+1$ homothetic unit triangles to cover any triangle larger than $T_n$. Note that Karabash and Soifer~\cite{KarabashSoiferGeom} showed that some shapes of large non-equilateral triangles could be covered with $n^2+1$ similar unit triangles, and this approach was extended to all shapes of non-equilateral triangles in~\cite{BoyerNonequilateralTriangles}. However, the methods in~\cite{BoyerNonequilateralTriangles, KarabashSoiferGeom} do not violate the proof by Baek and Lee~\cite{BaekLee} because the methods in~\cite{BoyerNonequilateralTriangles, KarabashSoiferGeom} complete the coverings using only one extra unit triangle that is similar to but not homothetic to the large triangle.

One aspect of the methods in~\cite{ConwaySoiferGeom, ConwaySoifer} is that they are part of existential proofs only, so as $n$ increases, there are vanishingly few large triangles on which they succeed. This is why the methods are described as covering triangles of the form $T_{n+\varepsilon}$, since $\varepsilon$ is a positive value that may become vanishingly small as $n$ increases. Recently, it was shown that one of the two methods by Conway and Soifer could be extended in a simple, natural way to use $n^2$ plus an even number of additional unit triangles, from 2 up to $2n$, to cover larger triangles with $d$ up to $\frac{1}{2}$~\cite{BoyerAlwaysCoverTriangles}. While this generalization could cover \emph{half} of all large triangles more efficiently than by the na\"{\i}ve covering method, it was left as an open question in~\cite{BoyerAlwaysCoverTriangles} to determine whether the method was optimal or if there exists a more efficient method. Similarly, after extending their proof technique to establish tight upper bounds for $d$ when using $n^2+2$ or $n^2+3$ homothetic unit triangles to cover a $T_{n+d}$, Baek and Lee stated as an open problem the need to find tight upper bounds for the $n^2 + k$ cases for $4 \le k \le 2n$~\cite{BaekLee}.

In this paper, we answer these two open questions from ~\cite{BaekLee, BoyerAlwaysCoverTriangles}. We begin in Section~\ref{sec:GeneralizedConway} with a review of the first method by Conway and Soifer as well as its generalization in~\cite{BoyerAlwaysCoverTriangles}. Then, in Section~\ref{sec:NewTriangleCoveringMethod_Even}, we present a new triangle covering method for the even cases of $k$, i.e., a new method for covering a $T_{n+d}$ with $n^2 + 2p$ homothetic unit triangles for each value $1 \le p \le n \in \mathbb{N}$. We show that the method is superior to the one in~\cite{BoyerAlwaysCoverTriangles} with a higher upper bound for $d$, culminating in the ability to \emph{almost always} cover triangles more efficiently than the na\"{\i}ve covering method using only $n^2 + 2n - 2c$ homothetic unit triangles for any constant $c \in \mathbb{N}_0$. In Section~\ref{sec:OptimalEvenCases}, we review the Baek and Lee proof for the $n^2+2$ case and extend it to establish upper bounds for all $n^2 + 2p$ cases with $1 \le p \le n \in \mathbb{N}$. We also prove that these upper bounds are tight by showing that they are met by the new triangle covering method for even cases. Next, we transition to the odd cases of $k$, beginning with a review in Section~\ref{sec:ThreeExtraUnitTriangles} of the Baek and Lee method and proof for the $n^2+3$ case. Although their method is optimal for the $n^2+3$ case, we observe that it generalizes to a good but not optimal method for any higher odd cases of $k$. In Section~\ref{sec:NewTriangleCoveringMethod_Odd}, we present a new triangle covering method for the odd cases of $k$, i.e., for covering a $T_{n+d}$ with $n^2 + 2p + 1$ homothetic unit triangles for each value $1 \le p < n \in \mathbb{N}$. Then, in Section~\ref{sec:OptimalOddCases}, we extend the Baek and Lee proof for the $n^2+3$ case to establish upper bounds for all $n^2 + 2p + 1$ cases with $1 \le p < n \in \mathbb{N}$, and we prove that these upper bounds are tight by showing that they are met by the new triangle covering method for odd cases. In Section~\ref{sec:ConsolidatedMethod}, we construct an optimal triangle covering method from the methods for the even and odd cases by defining how to decide which will use the fewest unit triangles to cover a given $T_{n+d}$. Finally, Section~\ref{sec:ConclusionFutureWork} presents concluding remarks and future work.

\section{Generalized First Conway-Soifer Method}\label{sec:GeneralizedConway}

The first Conway-Soifer triangle covering method in~\cite{ConwaySoiferGeom, ConwaySoifer} begins by placing a $T_n$ covering within a given $T_{n+d}$ so that the coverage gap appears below the $T_n$, as depicted in Figure~\ref{fig:CoverTriangleTn}(b). Then, the gap is covered by adding 2 extra unit triangles to the bottom trapezoidal row and using an \textbf{\emph{up-left slide operation}} on that row to rearrange the unit triangles. First, the row of unit triangles is placed along the base and the left diagonal of the $T_{n+d}$. Then, each down-triangle in the row is slid upward and leftward along the right-leg slope of its preceding up-triangle, and each succeeding unit triangle in the row is slid leftward to provide continuous coverage out to the right legs of the rightmost down- and up-triangles. The \textbf{\emph{slide value}} used in~\cite{ConwaySoiferGeom, ConwaySoifer} is $d = \varepsilon$ units, so the unit triangles cover a trapezoid with diagonal length of $1 + \varepsilon$ at the bottom of a $T_{n+\varepsilon}$~\cite[Figure 1]{ConwaySoifer}.

Due to the leftward sliding that ensures continuity of coverage, the original $2n-1$ unit triangles from the bottom trapezoidal row of a $T_n$ covering do not span the width of a $T_n$ much less a larger triangle. The 2 extra unit triangles increase the width of coverage. They are arranged in the same down-triangle then up-triangle configuration as the original $2n-1$ unit triangles. As shown for $n = 8$ in Figure 1 of~\cite{ConwaySoifer} and for $n = 12$ in Figure 1 of~\cite{ConwaySoiferGeom, Soifer}, the 2 extra unit equilateral triangles are sufficient to fully cover the given larger equilateral triangle, because $d = \varepsilon$ is small enough in those figures. The operation of the method is also shown in Figure~\ref{fig:CoverTriangleSideNPlus}, except $n$ and $d$ were selected to show that a larger $d$ value can cause a coverage gap, labeled $\alpha$, to occur.

\begin{figure}[!ht]
\centering

\includegraphics[width=3.6in,height=3in]{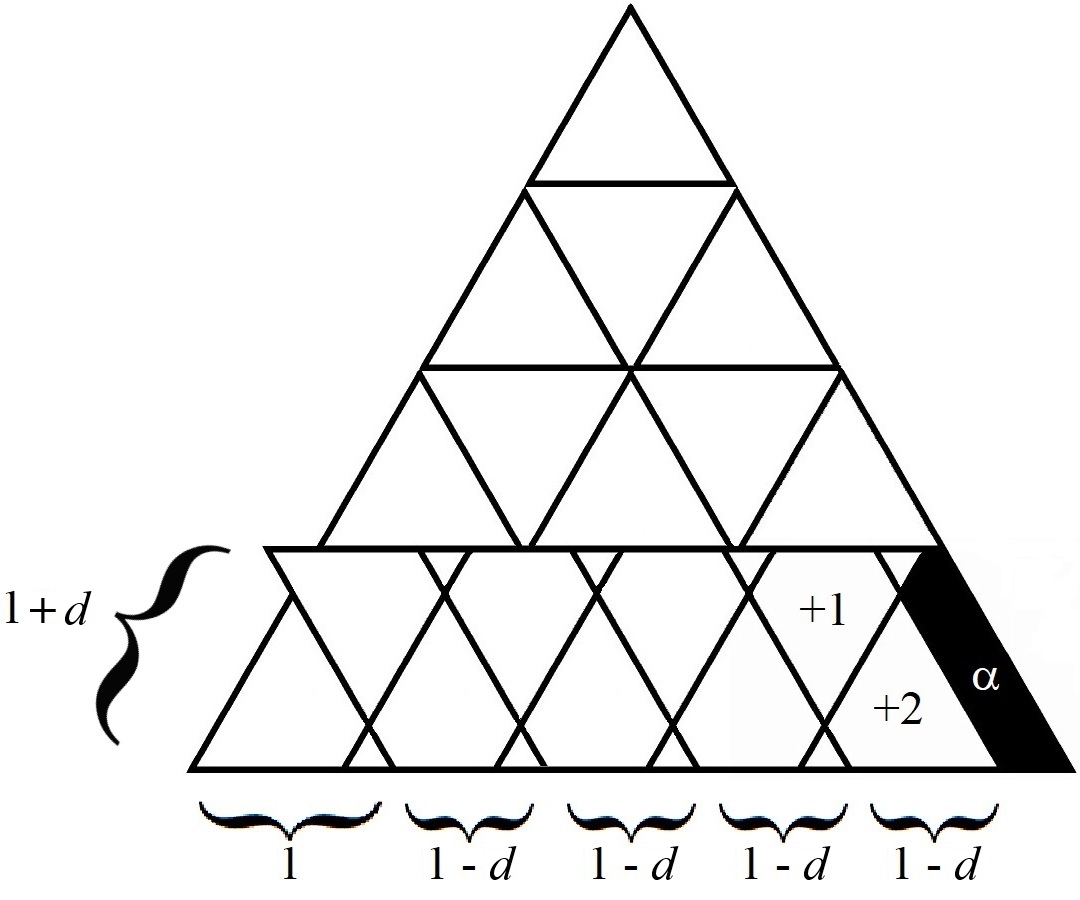}

\caption{The behavior of the first method in~\cite{ConwaySoiferGeom, ConwaySoifer} is shown using $n^2 + 2$ unit triangles to cover all of a $T_{n + d}$ except for a coverage gap labeled $\alpha$ that can occur if $d$ is too large.
\label{fig:CoverTriangleSideNPlus}}
\end{figure}

When the bases of the $n+1$ up-triangles in the bottom trapezoidal row are arranged as depicted in Figure~\ref{fig:CoverTriangleSideNPlus} (and Figure 1 of~\cite{ConwaySoiferGeom, ConwaySoifer, Soifer}), then they can span the base of a $T_{n + d}$ if $d$ is small enough. However, each up-triangle after the first only contributes a width of $(1 - d)$, so the method only fully covers $T_{n + d}$ when $n + d \le 1 + n (1 - d)$. By subtracting $n$ from both sides, we get $d \le 1 - nd$. Solving for $d$ yields $d \le \frac{1}{n + 1}$.

Hence, as $n$ increases, this method covers vanishingly few $T_{n+d}$ with $d \in (0, 1)$. A generalization of this method was presented and analyzed in~\cite{BoyerAlwaysCoverTriangles}. Instead of adding only one pair of unit triangles, the generalized method continues adding down- and up-triangle pairs into the bottom trapezoidal row until they span the base of a given $T_{n+d}$. The number of additional pairs of unit triangles to add is determined based on the value of $d$. For example, the generalization uses only $n^2+4$ unit triangles as long as $d \le \frac{2}{n + 2}$. However, when $d$ is large enough that the generalization would add more than $2n$ extra unit triangles to the bottom row, then instead it reverts to using the na\"{\i}ve covering method, which the analysis in~\cite{BoyerAlwaysCoverTriangles} showed is only necessary if $d > \frac{1}{2}$.

\section{New Triangle Covering Method (Even Cases)}\label{sec:NewTriangleCoveringMethod_Even}

The open question from~\cite{BoyerAlwaysCoverTriangles}: For those $T_{n+d}$ in which $d > \frac{1}{2}$, can we do better than the na\"{\i}ve covering method? The intuition that there may be a better solution is based on a parallel to Soifer's originating question (in Section~\ref{sec:Intro}): Rather than a given triangle being a little larger than a $T_n$, as in Figure~\ref{fig:CoverTriangleTn}(b), if a given triangle is smaller than a $T_n$ by a little \emph{more} than the area of a unit triangle, then can it be covered with one \emph{less} than $n^2$ unit triangles? In this section, we answer in the affirmative. First, we present the main idea of the new covering method, which is more efficient than the na\"{\i}ve covering method on \emph{almost all} triangles, rather than only \emph{half} as in the method of~\cite{BoyerAlwaysCoverTriangles}. Next, we show how to use fewer extra pairs of unit triangles as $d$ decreases, and we show that, as $n$ increases, \emph{almost all} triangles with selected side length between $n$ and $n+1$ can be covered by the new method with at most $n^2 + 2n - 2c$ homothetic unit triangles for any constant $c \in \mathbb{N}_0$. Then, we show how to select the minimum number of extra pairs of unit triangles to use in covering a $T_{n+d}$ based on the values of $n$ and $d$.

\subsection{The Main Idea (Even Cases)}\label{subsec:MainIdeaMethod_Even}

For a given equilateral triangle $T_{n+d}$ with $d \in (0, 1)$, we will begin by always using $n^2 + 2n$ unit triangles, since this is one less than the na\"{\i}ve covering method uses. Rather than adding pairs of unit triangles into only the last trapezoidal row to cover the $\alpha$ gap, the main idea of the new covering method for even cases consists of the following two elements:

\vspace{0.15in}

\begin{enumerate}
  \item For each row $k$, add \emph{only one} extra down-triangle and up-triangle pair.
  \item For each row $k$, perform the up-left slide operation on the unit triangles using a slide value of $\frac{1}{k(k+1)}$.
\end{enumerate}

An example of this method in operation appears in Figure~\ref{fig:NewTriangleCoverMethod}. The top row, which is row 1, has only one up-triangle from the $T_n$ covering, plus the 2 extra unit triangles for the row. In each row, the addition of the 2 extra unit triangles enables us to apply the up-left slide operation that increases the diagonal length of the triangle or trapezoid that is continuously covered by the unit triangles in the row. Rather than using a slide value of $\varepsilon$, the slide value $s_k$ to use in each row $k$ is calculated to be $\frac{1}{k(k+1)}$. For example, the slide value in the top row, denoted $s_1$, is calculated to ensure that the amount of extra coverage added by increasing the diagonal length of the row is equal to the width of the base of the row after adding the 2 extra unit triangles. The base of row 1 has width $2 - s_1$ because it is lined by the 2 up-triangles, but they overlap by the amount of the slide value. Solving $1+s_1 = 2-s_1$ yields $s_1 = \frac{1}{2}$. Similarly, in Figure~\ref{fig:NewTriangleCoverMethod}, the slide value in the second row, denoted $s_2$, is shown to be $\frac{1}{6}$, which was obtained by solving for $s_2$ in $1 + s_1 + 1 + s_2 = 3 - 2s_2$, given that $s_1 = \frac{1}{2}$.

\begin{figure}[!ht]
\centering
\includegraphics[width=4.7in,height=2.4in]{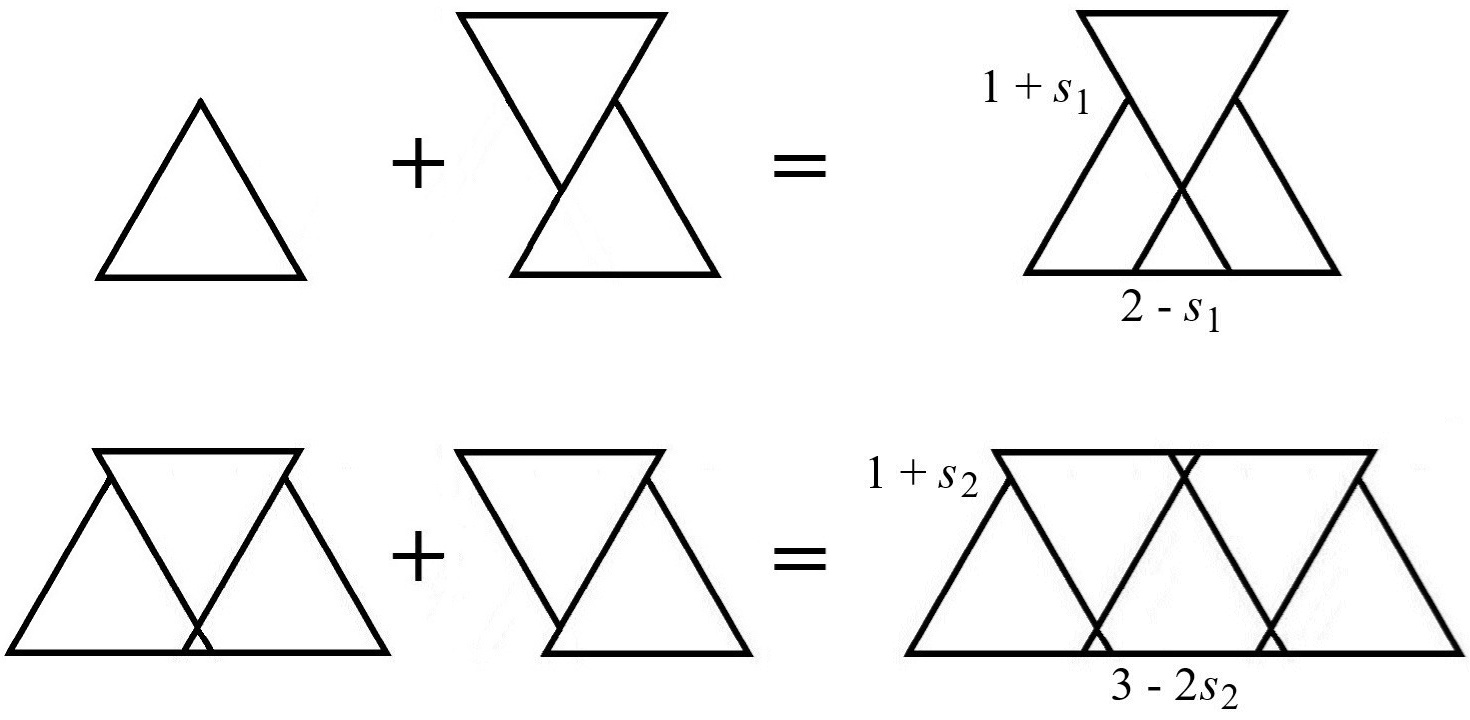}

\caption{Depiction of the main idea of the new method for covering a large triangle $T_{n+d}$ with extra pairs of homothetic unit triangles. This depiction has $n=2$, so it works for $d$ up to $\frac{2}{3}$.
\label{fig:NewTriangleCoverMethod}}
\end{figure}

For a given $T_{n+d}$, if the value of $d$ exceeds $\sum_{k=1}^{n} s_k$, then it is necessary to revert to the na\"{\i}ve covering method, i.e., covering the $T_{n+d}$ with a $T_{n+1}$ covering. However, this method can already be seen as improving on the Generalized First Conway-Soifer Method in~\cite{BoyerAlwaysCoverTriangles} (and Section~\ref{sec:GeneralizedConway}) because even at $n=2$, it has a higher threshold for $d$ before reverting to the na\"{\i}ve covering method, and the threshold for $d$ monotonically increases as $n$ increases. For example, continuing the calculations for $s_3$ and $s_4$ yields slide values of $\frac{1}{12}$ and $\frac{1}{20}$, so the maximum value for $d$ is $0.8$ for $n=4$.

\subsection{Using Fewer Pairs of Unit Triangles (Even Cases) }\label{subsec:UsingFewerUnitTriangles_Even}

As presented in the preceding subsection, the new covering method always uses $n^2 + 2n$ homothetic unit triangles to cover a $T_{n+d}$. We now adjust the new covering method so that it only starts adding pairs of extra homothetic unit triangles on a given row $j$:

\begin{enumerate}
  \item Let rows 1 to $j-1$ of the covering be the covering of a $T_{j-1}$ (no extra down-triangle and up-triangle pairs and no up-left slide operations).
  \item For each row $k$ from $j$ to $n$, add an extra down-triangle and up-triangle pair.
  \item For each row $k$ from $j$ to $n$, perform the up-left slide operation on its unit triangles using a slide value of $\frac{j}{k(k+1)}$.
\end{enumerate}

To facilitate exposition, the term \textbf{\textsc{CoverTriangle} (Basic Even)} refers to the adaptation of the main idea that incorporates the three elements above. Note that the method may also revert to the na\"{\i}ve covering method for a $T_n$ if $d = 0$ or for a $T_{n+1}$ if $d$ is too large according to the threshold in Theorem~\ref{thm:CoverEqTriangle_Even_Basic}, which establishes the efficacy of the main three elements of \textsc{CoverTriangle} (Basic Even). Also implicit in the method are the affine transformations that can be applied, before and after the main three elements above, to enable covering of non-equilateral triangles.

\begin{theorem}\label{thm:CoverEqTriangle_Even_Basic}
Given a triangle $T_{n+d}$ having a selected side length of $n+d$, with $n \in \mathbb{N}$ and $d \in (0, 1)$, and a row $1 \le j \le n \in \mathbb{N}$, \textsc{CoverTriangle} (Basic Even) covers the $T_{n+d}$ if $d \le \sum_{k=j}^{n} s_k$.
\end{theorem}
\begin{proof}
We focus on the equilateral $T_{n+d}$ (after affine transformation, if needed). First consider the rows above $j$. Since sliding is not applied to the rows 1 to $j-1$, nor are the extra unit triangles added, rows 1 to $j-1$ form an equilateral $T_{j-1}$ covering. The width of the base of a $T_{j-1}$ is $j-1$. Row $j$ is the first to receive 2 extra equilateral unit triangles and to have the up-left slide operation applied. The slide value on row $j$ is denoted $s_j$. The combination of the equilateral $T_{j-1}$ and row $j$ must satisfy the equation $j+s_j = j + 1 - js_j$ because the diagonal length must equal the base length of the equilateral $T_{j+s_j}$ that is being covered by the combination. Solving the equation for $s_j$ yields $s_j = \frac{1}{j + 1}$. The top of row $j$ is lined by $j$ down-triangles, and each down-triangle except the last overlaps its succeeding down-triangle by $s_j$ units. Therefore, the total width of the top of row $j$ is $j - (j-1)s_j$. However, small parts of the leftmost and rightmost down-triangles spill over (do not cover part of the area of) the $T_{j+s_j}$ that is being covered, namely small equilaterally triangular parts with side lengths of $s_j$. Therefore, the portion of the top of row $j$ that intersects with the $T_{j+s_j}$ that is being covered has length $j - (j+1)s_j$. Substituting the previously calculated value of $s_j = \frac{1}{j + 1}$ yields a length of $j-1$, which is the exact width of the base of the $T_{j-1}$ covering generated for rows 1 to $j-1$.

Given that the top $j$ rows fit together into a full covering of a $T_{j+s_j}$, we turn to the rest of the covering of an equilateral $T_{n+d}$. We use an inductive proof in which the unit equilateral triangle slide value of $\frac{j}{k(k+1)}$ used on each row $k \ge j$ gives row $k$ a base length equal to the length of the diagonal formed by stacking rows 1 to $k$ in the shape of an equilateral $T_{k + \sum_{i=j}^{k} s_i}$. As the inductive anchor, we have this property for $k=j$ because the assignment $j+s_j = j + 1 - js_j$ yielded $s_j = \frac{1}{j + 1}$. Multiplying by $\frac{j}{j}$ and then substituting $j=k$ in the denominator yields the anchor case of $s_j = \frac{j}{k(k + 1)}$. As the inductive hypothesis, we assume that the base length equals the diagonal length for each row $i$ from $j$ to $k-1$ by assigning each $s_i = \frac{j}{i(i+1)}$. As the inductive step, we prove that assigning $s_k = \frac{j}{k(k+1)}$ yields a base length for row $k$ that is equal to the diagonal length of a $T_{k+\sum_{i=j}^{k} s_i}$.

The diagonal length of a $T_{k+\sum_{i=j}^{k} s_i}$ is equated to its base length as follows:

\begin{eqnarray*}
k + \sum\limits_{i=j}^{k} s_i = k + 1 - k \times s_k
\end{eqnarray*}

\noindent Simplifying both sides, we obtain:

\begin{eqnarray*}
\sum\limits_{i=j}^{k} s_i = 1 - k \times s_k
\end{eqnarray*}

\noindent Next, we substitute known values from the inductive hypothesis:

\begin{eqnarray*}
s_k + \sum\limits_{i=j}^{k-1} \frac{j}{i(i+1)} = 1 - k \times s_k
\end{eqnarray*}

\noindent Then we perform partial fraction decomposition on the summation term:

\begin{eqnarray*}
s_k + j \times \sum\limits_{i=j}^{k-1} \left(\frac{1}{i} - \frac{1}{i+1}\right) = 1 - k \times s_k
\end{eqnarray*}

\noindent Which we can telescopically collapse and rearrange to:

\begin{eqnarray*}
(k+1) s_k &=& 1 - j \times \left(\frac{1}{j} - \frac{1}{k}\right) = \frac{j}{k}
\end{eqnarray*}

\noindent Such that:

\begin{eqnarray*}
s_k &=& \frac{j}{k(k+1)}
\end{eqnarray*}

\noindent Hence, the diagonal length of $T_{k+\sum_{i=j}^{k} s_i}$ is equal to the base length of row $k$ when we assign $s_k = \frac{j}{k(k+1)}$, which completes the inductive proof that \textsc{CoverTriangle} (Basic Even) fully covers any equilateral triangle $T_{n+d}$ for which $d \le \sum_{k=j}^{n} s_k$. The theorem follows for all triangles via affine transformations.
\end{proof}

In \textsc{CoverTriangle} (Basic Even), the slide value is magnified (multiplied) by $j$, the row on which we start adding the extra pair of unit triangles and performing the up-left slide operation. As an example, if we recalculate the values of $s_1$ and $s_2$ with $j=2$, $s_1$ is 0 rather than $\frac{1}{2}$, but $s_2 = \frac{1}{3}$ rather than only $\frac{1}{6}$. As shown in Theorem~\ref{thm:CoverEqTriangle_Even_Basic_AlmostAlwaysBetter}, even though the series of non-zero slide values starts with a smaller fraction, the sum of the slide values converges to 1 for any constant $j$. For example, with $j=1$, the series contains the reciprocals of the \textbf{pronic numbers} and with $j=2$, the series contains the fractional reciprocals of the \textbf{triangular numbers}, both of which converge to 1.

\begin{theorem}\label{thm:CoverEqTriangle_Even_Basic_AlmostAlwaysBetter}
With $n \in \mathbb{N}$ and $d \in (0, 1)$, triangles having selected side lengths of $n+d$ can \emph{almost always} be covered using at most $n^2 + 2n - 2(j-1)$ homothetic unit triangles for any constant $j \in \mathbb{N}$.
\end{theorem}
\begin{proof}
Given any such triangle $T_{n+d}$, \textsc{CoverTriangle} (Basic Even) can cover the $T_{n+d}$ if $d \le \sum_{k=j}^{n} s_k$, per Theorem~\ref{thm:CoverEqTriangle_Even_Basic}.
This threshold for $d$ is the multiple $j$ of the sum of the reciprocals of the pronic numbers starting at term $j$, which converges to 1 as follows:

\begin{eqnarray*}
\sum_{i=j}^{\infty} s_i &=& \sum_{i=j}^{\infty} \frac{j}{i(i+1)}\\
                        &=& j \times \left( \sum_{i=1}^{\infty} \frac{1}{i(i+1)} - \sum_{i=1}^{j-1} \frac{1}{i(i+1)}\right)\\
                        &=& j \times \left( 1 - \sum_{i=1}^{j-1} \left( \frac{1}{i} - \frac{1}{i+1} \right) \right)\\
                        &=& j \times \left( 1 - \left( 1 - \frac{1}{j} \right) \right)\\
                        &=& 1
\end{eqnarray*}

\noindent Thus, given a randomly selected value for $d$, as $n$ increasingly exceeds any selected constant $j \in \mathbb{N}$, the probability approaches 1 that \textsc{CoverTriangle} (Basic Even) will cover an equilateral $T_{n+d}$ with at most $n^2 + 2n - 2(j-1)$ homothetic unit equilateral triangles. The theorem follows for all triangles via affine transformations.
\end{proof}

\subsection{Selecting the Best \texorpdfstring{$j$}{\emph{j}} (Even Cases)}\label{subsec:SelectingJ_Even}

It is possible to poorly select the row $j$ on which to start adding the extra unit triangle pairs and performing the up-left slide operations. For example, given a $T_{n+d}$ with $d = \frac{3}{4}$, if $j$ were selected to be $n$, then \textsc{CoverTriangle} (Basic Even) would not generate a large enough covering. This is of little consequence in the proof of Theorem~\ref{thm:CoverEqTriangle_Even_Basic_AlmostAlwaysBetter} because $n$ increases arbitrarily, and $j$ is a constant. However, when given a particular $T_{n+d}$ to cover, it is both necessary and possible to use the values of $n$ and $d$ to calculate the best row $j$ to select that minimizes the number of additional pairs of unit triangles that are added by \textsc{CoverTriangle} (Basic Even). Specifically, Theorem~\ref{thm:CoverEqTriangle_SelectingJ_Even} characterizes the maximum value of $j$ that can be selected to ensure that \textsc{CoverTriangle} (Basic Even) covers a given $T_{n+d}$.

\begin{theorem}\label{thm:CoverEqTriangle_SelectingJ_Even}
A triangle $T_{n+d}$ having a selected side length of $n+d$, with $n \in \mathbb{N}$ and $d \in (0, 1)$, can be covered by $n^2 + 2n - 2(j-1)$ homothetic unit triangles with $j = \lfloor (1 - d)(n + 1) \rfloor$ and not with $j = \lfloor (1 - d)(n + 1) \rfloor + 1$.
\end{theorem}
\begin{proof}
The value of $d$ must not exceed the sum of the slide values used on rows $j$ through $n$ of the covering that the \textsc{CoverTriangle} (Basic Even) method generates. Hence, we start with that inequality constraint for equilateral triangles and then solve for $j$, as follows:

\begin{eqnarray*}
d &\le& \sum_{i=j}^{n} s_i\\
  &=& \left( \sum_{i=1}^{n} \frac{j}{i(i+1)} \right) - \left( \sum_{i=1}^{j-1} \frac{j}{i(i+1)} \right)\\
  &=& j \times \left( \sum_{i=1}^{n} \left( \frac{1}{i} - \frac{1}{i+1} \right)  - \sum_{i=1}^{j-1} \left( \frac{1}{i} - \frac{1}{i+1} \right) \right)\\
  &=& j \times \left( \left( 1 - \frac{1}{n+1} \right) - \left( 1 - \frac{1}{j} \right) \right)\\
  &=& 1 - \left( \frac{j}{n+1} \right)
\end{eqnarray*}

\noindent Solving for $j$, we find that we must select a row such that:

\begin{eqnarray*}
j \le (1 - d)(n + 1)
\end{eqnarray*}

Since $j$ must be an integer, we use the floor function to ensure that the selected row $j$ is still small enough produce a sum of slide values that is equal to or exceeds the value of $d$. Therefore, we select:

\begin{eqnarray*}
j = \lfloor (1 - d)(n + 1) \rfloor
\end{eqnarray*}

Next, we observe that the value $\lfloor (1 - d)(n + 1) \rfloor + 1$ cannot be selected because it is greater than $(1 - d)(n + 1)$. Furthermore, for the vanishingly few triangles (per Theorem~\ref{thm:CoverEqTriangle_Even_Basic_AlmostAlwaysBetter}) that $j$ is assigned the value 0, the $T_{n+d}$ is covered by reverting to the na\"{\i}ve covering method. The theorem follows for all triangles via affine transformations.
\end{proof}

To facilitate exposition, the term \textbf{\textsc{CoverTriangle} (Full Even)} refers to the version of the \textsc{CoverTriangle} (Basic Even) that selects $j$ using $j = \lfloor (1 - d)(n + 1) \rfloor$. As an example of the theorem, with $n = 4$ and $d = 0.3$, $j$ is selected to be the floor of $\frac{7}{2}$, which is 3.

\section{Proof of Optimality for the Even Cases}\label{sec:OptimalEvenCases}

In this section, we first review the Baek and Lee proof for the $n^2+2$ case and then extend it to establish upper bounds for all $n^2 + 2p$ cases with $1 \le p \le n \in \mathbb{N}$. Then, we prove that these upper bounds are tight by showing that they are met by \textsc{CoverTriangle} (Full Even).

\subsection{The Baek-Lee Machinery}\label{subsec:BaekLeeMachinery}

From~\cite{BaekLee}, an \textbf{\emph{$H$-triangle}} is a triangle $T$ with one leg $l$ parallel to the $x$-axis. For an $H$-triangle $T$, the $y$-coordinate of leg $l$ is denoted $y_{_T}$. The \textbf{\emph{base}} of $T$ is the length of $l$, and the \textbf{\emph{height}} of $T$ is the minimal distance between the vertex of $T$ not in $l$ and the line containing $l$. A unit $H$-triangle has a base and height of 1 unit. Without loss of generality, any triangle $T$ can be converted to an $H$-triangle by first rotating it to have a selected side parallel to the $x$-axis, as its leg, and positioning it with the left vertex of its leg at the origin. Then, coordinate scaling can be performed so that its height matches its base. Hence, in this section, we let $T_{n+d}$ denote an $H$-triangle to be covered, let $T_1$ denote a unit $H$-triangle, let \textbf{\emph{up-triangle}} refer to an $H$-triangle with its leg at the least $y$-coordinate it occupies, and let \textbf{\emph{down-triangle}} refer to an $H$-triangle with its leg at the greatest $y$-coordinate it occupies.

To enable analysis of what is feasible for any triangle covering method, Baek and Lee defined functions to help immunize measurement of the areas of triangles from their exact positions in a covering. First was a utility function $\tilde f_{_T}(t)$ for the width of the intersection between an $H$-triangle $T$ and the line segment $y=t$. The second was a function $f_{_T}(t): [0, 1)\to\mathbb{R}$ to collapse the domain relative to $\tilde f_{_T}(t)$, as follows: let $f_{_T}(t) = \sum_{n \in \mathbb{Z}} \tilde f_{_T}(t+n)$.

Let $\{x\}$ denote $x$ modulo 1, e.g., $\{2.4\} = 0.4$ and $\{-2.4\} = 0.6$. For a unit $H$-triangle, $f_{_{T_1}}(t) = \{t - y_{_{T_1}}\}$ if the $T_1$ is a down-triangle, and $f_{_{T_1}}(t) = 1 - \{t - y_{_{T_1}}\}$ if the $T_1$ is an up-triangle. Observe that these are equivalent to $f_{_{T_1}}(t) = \{t - a\}$ for a down-triangle and $f_{_{T_1}}(t) = 1 - \{t - a\}$ for an up-triangle with $a$ being the difference between the highest $y$-coordinate of the $T_1$ and the least integer $k$ of a line $y=k$ that intersects the $T_1$. For any unit $H$-triangle, the function $f_{_{T_1}}(t)$ measures the change in width of the $T_1$ as $t$ increases, so it is a linear function with a discontinuity at $a$ and a slope of -1 for an up-triangle or 1 for a down-triangle.

Important aspects of the coverings of large $H$-triangles with unit $H$-triangles can be modelled by functions in an abelian group $\mathcal{T}$ of ``all functions from $[0, 1) \to \mathbb{R}$ that can be expressed as a finite addition and subtraction of generating functions of the form $t \mapsto \{t - a\}$ and $t \mapsto 1 - \{t - a\}$ with arbitrary $a \in [0, 1)$"~\cite[p. 117]{BaekLee}. Each such function $f \in \mathcal{T}$ has a number of properties in common, such as being linear with a slope $m \in \mathbb{Z}$ corresponding to the difference in positive and negative $t$'s in $f$'s finite addition and subtraction series of the generating functions. The area under any function $f \in \mathcal{T}$ is denoted $\int f$. Furthermore, because $\int f$ is measuring the area of a triangle with an equal base and height, it equals $\frac{b}{2}$ for some $b \in \mathbb{Z}$ where $b - m$ is divisible by 2. For example, given the equilateral triangle covering in Figure~\ref{fig:CoverTriangleTn}(a), stretch the $y$-coordinate system by $\frac{2}{\sqrt{3}}$ and then let $f$ represent the unit triangles in the covering. Then, $\int f = \frac{n^2}{2} = \frac{9}{2}$, and the slope of $f$ is based on having one more up-triangle (one extra $1 - \{t - a\}$) per row of the $T_n$ covering.

If both a $T_{n+d}$ and a $T_n$ are positioned as upward pointing with their base legs on the $x$-axis and their left base-leg vertices at the origin, then $f_{_{T_{n+d}}}(t)$ and $f_{_{T_{n}}}(t)$ measure the sums of the widths of the $T_{n+d}$ and the $T_n$ take at $t$ units above each integer from $0$ through $n$. At every one of the $n$ vertical positions from $0+t$ to $(n-1)+t$, the $T_{n+d}$ is wider by $d$ than the $T_n$, and at $n+t$, the $T_n$ has width $0$ and the $T_{n+d}$ has width $max(0, d-t)$. Accordingly, for all $t \in [0, 1)$, we have:

\begin{eqnarray*}
f_{_{T_{n+d}}}(t) = f_{_{T_{n}}}(t) + nd + max(0, d-t)
\end{eqnarray*}

Baek and Lee used this machinery to help prove that the first Conway-Soifer method in~\cite{ConwaySoiferGeom, ConwaySoifer} was an optimal method for covering a $T_{n+d}$ when restricted to only $n^2+2$ homothetic unit triangles. First, they showed that it can cover a $T_{n+d}$ with no $\alpha$ gap using only $n^2+2$ homothetic unit triangles if $d \le \frac{1}{n+1}$ (see Section~\ref{sec:GeneralizedConway}). Then, they showed that there was a contradiction to the feasibility of fully covering a $T_{n+d}$ with only $n^2+2$ homothetic unit triangles if $d > \frac{1}{n+1}$.

\begin{theorem}[Baek and Lee~\cite{BaekLee}]\label{thm:BaekLee_Even}
If $d > \frac{1}{n+1}$, then a triangle $T_{n+d}$ having a selected side length of $n+d$ cannot be covered by $n^2 + 2$ homothetic unit triangles.
\end{theorem}
\begin{proof}
Affine-transform the $T_{n+d}$ into an $H$-triangle. Let $N = n^2 + 2$ and let $S_1, \ldots, S_N$ be the homothetic unit $H$-triangles that are assumed to cover the $T_{n+d}$ of base and height $n+d$. Let $f{_{S}} = \sum_{i=1}^{N} f_{_{S_i}}$. Define the area reference function $r \in \mathcal{T}$ to be $r = f_{_{T_{n}}} + 1$. Note that $\int r = \frac{n^2 + 2}{2}$ because a $T_n$ has area $\frac{n^2}{2}$ and two extra unit $H$-triangles each add $\frac{1}{2}$ to the area. Because the domain of functions in $\mathcal{T}$ is $[0, 1)$, the increase of area by 1 has the effect of increasing $r$ by 1 relative to $f_{_{T_{n}}}$.

To enable comparison of the area occupied by a $T_{n+d}$ with that of the assumed covering provided by $S_1, \ldots, S_N$, define $g = f_{_{T_{n+d}}} - r$ and $h = f_{_{S}} - r$. For all $t \in [0, 1)$, we have:

\begin{eqnarray*}
g(t) &=& f_{_{T_{n+d}}}(t) - (f_{_{T_{n}}}(t) + 1)\\
     &=& (f_{_{T_{n}}}(t) + nd + max(0, d-t)) - (f_{_{T_{n}}}(t) + 1)\\
     &=& nd + max(0, d-t) - 1\\
     &=& -(1 - nd) + max(0, d-t)\\
     &\ge& -(1 - nd) + d - t\\
     &\ge& -t + ((n+1)d - 1)
\end{eqnarray*}

Let $\delta = ((n+1)d - 1)$. Then, observe that $\delta > 0$ when $d > \frac{1}{n+1}$. Given the definition $h = f_{_{S}} - r$, we can conclude that $\int h \ge \frac{1}{2}$ according to Corollary 16 in~\cite{BaekLee} as follows. By the contradictive assumption (that $S_1, \ldots, S_N$ cover the $T_{n+d}$), we have $h(t) \ge g(t)$ for all $t \in [0, 1)$. Define $g^{\prime} = -t$. Note that $1 \in \mathcal{T}$ because $1 = (1 - \{t - a\}) + \{t - a\}$ and therefore $g^{\prime} \in \mathcal{T}$ because $-t = -\{t\}$ in the domain $t \in [0, 1)$, and $-\{t\} = (1 - \{t - a\}) + 1$ with $a = 0$. Hence, $h(t) \ge g(t) = g^{\prime}(t) + \delta$ for all $t \in [0, 1)$. Define $h^{\prime} = h - g^{\prime} \in \mathcal{T}$. The function $h^{\prime}$ must be positive for all $t \in [0, 1)$ because $h^{\prime} = h - g^{\prime} \ge g - g^{\prime} \ge \delta$ and $\delta > 0$ when $d > \frac{1}{n+1}$. Since $h^{\prime}(t)$ is always positive, $\int h^{\prime} > 0$. According to Lemma 14 in~\cite{BaekLee}, $\int h^{\prime}$ must also be a half-integer or integer, because each function in $\mathcal{T}$ is a series of additions and subtractions of functions representing unit $H$-triangles that each have an area of $\frac{1}{2}$. Yet we must also have $\int h^{\prime} > \frac{1}{2}$ due to Lemma 15 in~\cite{BaekLee}. Specifically, any function $f \in \mathcal{T}$ such that $\int f = \frac{1}{2}$ must of the form $t \mapsto \{mt + c\}$ or $t \mapsto 1 - \{mt + c\}$ for a positive integer $m$ and $c \in [0, 1)$. Both functional forms have an $x$-intercept, but $h^{\prime}(t)$ is always positive for all $t \in [0, 1)$. Hence, $\int h^{\prime} \ge 1$. Since $g^{\prime} = -t$, $\int g^{\prime} = -\frac{1}{2}$. Furthermore, since $h^{\prime} = h - g^{\prime}$ by definition, we have $h = h^{\prime} + g^{\prime}$, so we must have:

\begin{eqnarray*}
\int h = \int h^{\prime} + \int g^{\prime} \ge 1 + -\frac{1}{2} \ge \frac{1}{2}
\end{eqnarray*}

\noindent However, the function $f{_{S}}$ is defined to be equal to $\sum_{i=1}^{N} f_{_{S_i}}$ with $N = n^2 + 2$, and $r = f_{_{T_{n}}} + 1$, so $r$ also represents $n^2 + 2$ unit $H$-triangles. Hence, since $h = f_{_{S}} - r$, we must have:

\begin{eqnarray*}
\int h = \int f_{_{S}} - \int r = 0
\end{eqnarray*}

\noindent Thus, we have reached a contradiction, which proves the constraint $d \le \frac{1}{n+1}$ for the feasibility of covering a $T_{n+d}$ using only $n^2 + 2$ homothetic unit triangles.
\end{proof}

\subsection{Optimality in the Even Cases}\label{subsec:OptimalEvenCases}

In Theorem~\ref{thm:CoverEqTriangle_Optimal_Even}, we prove that \textsc{CoverTriangle} (Full Even) successfully covers a $T_{n+d}$ with $n^2 + 2p$ homothetic unit triangles for each $1 \le p \le n \in \mathbb{N}$ if $d \le \frac{p}{n+1}$, and we prove that it is infeasible to do so if $d > \frac{p}{n+1}$.

\begin{theorem}\label{thm:CoverEqTriangle_Optimal_Even}
\textsc{CoverTriangle} (Full Even) is an optimal method for covering a $T_{n+d}$ with $n^2 + 2p$ homothetic unit triangles for each $1 \le p \le n \in \mathbb{N}$.
\end{theorem}
\begin{proof}
To generalize the proof of Theorem~\ref{thm:BaekLee_Even}, we let $N = n^2 + 2p$ for $1 \le p \le n \in \mathbb{N}$ so that the $S_1, \ldots, S_N$ unit $H$-triangles assumed to cover a $T_{n+d}$ are $n^2 + 2p$ in number. Next, we amend the definition of the reference function $r$ to be $r = f_{_{T_{n}}} + p$ such that $r \in \mathcal{T}$ and $\int r = \frac{n^2 + 2p}{2}$ because a $T_n$ has area $\frac{n^2}{2}$ and each extra pair of unit $H$-triangles adds $1$ to the area. This change to $r$ has the effect of substituting $p$ for $1$ in the relevant places of the proof of Theorem~\ref{thm:BaekLee_Even}. In particular, we have the following amended calculations for $g$:

\begin{eqnarray*}
g(t) &=& f_{_{T_{n+d}}}(t) - (f_{_{T_{n}}}(t) + p)\\
     &\ge& -t + ((n+1)d - p)
\end{eqnarray*}

\noindent such that the $\delta > 0$ needed for applicability of Corollary 16 in~\cite{BaekLee} is based on the assignment $\delta = ((n+1)d - p)$. Hence, we have the following upper bound for $d$, above which the proof of Theorem~\ref{thm:BaekLee_Even} establishes a contradiction to the feasibility of covering a $T_{n+d}$ with $n^2 + 2p$ homothetic unit triangles:

\begin{eqnarray*}
d > \frac{p}{n+1}
\end{eqnarray*}

From the proof of Theorem~\ref{thm:CoverEqTriangle_SelectingJ_Even}, \textsc{CoverTriangle} (Full Even) has the following upper bound for $d$:

\begin{eqnarray*}
d \le 1 - \frac{j}{n+1}
\end{eqnarray*}

\noindent Observing that $p = n - j + 1$ and simplifying, we have that \textsc{CoverTriangle} (Full Even) can cover a $T_{n+d}$ with $n^2 + 2p$ homothetic unit triangles when:

\begin{eqnarray*}
d \le \frac{p}{n+1}
\end{eqnarray*}

Thus, the upper bounds in this proof are tight, and \textsc{CoverTriangle} (Full Even) is an optimal covering method when adding pairs of homothetic unit triangles. The theorem holds for all triangles via affine transformations.
\end{proof}

\section{The Smallest Odd Case}\label{sec:ThreeExtraUnitTriangles}

Now, we transition to considering the odd cases, i.e., the cases of covering a $T_{n+d}$ with $n^2 + 2p + 1$ homothetic unit triangles. We return to our focus on equilateral triangles, except when we must transform to $H$-triangles inside the proofs based on the Baek and Lee proof technique. In fact, even Baek and Lee present their \emph{method} for the $n^2+3$ case using equilateral triangles. We observe that although their method for the $n^2+3$ case is optimal, the method does not generalize to an optimal method for any higher odd cases, when $p > 1$. Their main idea is to amend the behavior of the first Conway-Soifer method to add 3 homothetic unit triangles to the last trapezoidal row, rather than adding only 2. An example of this technique appears in Figure~\ref{fig:BaekLee_OddCases}. In the $n^{th}$ row, the up-left slide operation is still performed but only on the first $n-1$ down-triangles because the last down-triangle and the 3 extra unit triangles form a $T_2$ with no sliding.

\begin{figure}[!ht]
\centering
\includegraphics[width=3.3in,height=1.5in]{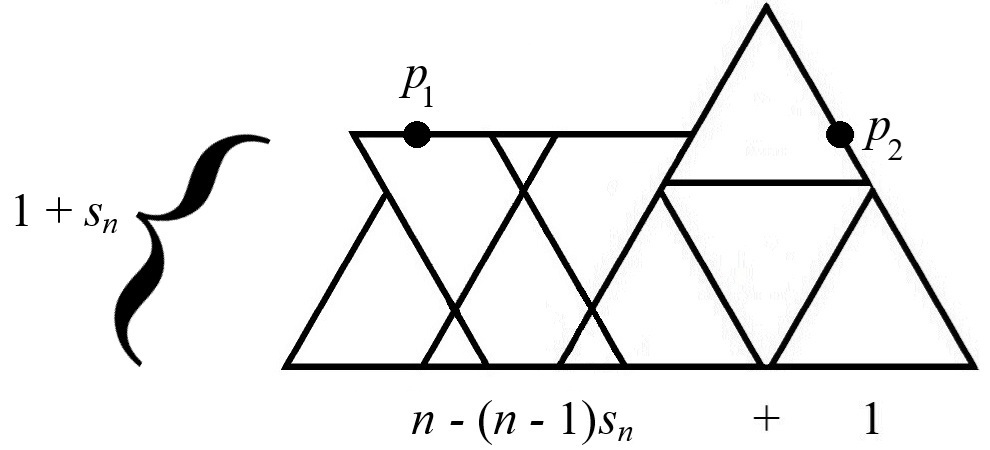}

\caption{This figure shows the structure of the $n^{th}$ (bottom) row when 3 extra unit triangles are appended, instead of just the 2 that are added in the first Conway-Soifer method. For a small enough value of $d$, a $T_{n+d}$ covering is completed by placing a $T_{n-1}$ atop this structure. For example, in this depiction, $n = 3$, so a $T_2$ would be positioned with its base spanning from $p_1$ to $p_2$.
\label{fig:BaekLee_OddCases}}
\end{figure}

For the $n^2+3$ case, there are no extra unit triangles to add to the rows above $n$, so a $T_{n-1}$ covering is placed above the structure in Figure~\ref{fig:BaekLee_OddCases} with its left diagonal aligned with that of the leftmost up-triangle in Figure~\ref{fig:BaekLee_OddCases} and its base overlapping horizontal leg of the leftmost triangle starting at point $p_1$ and extending to the point $p_2$ in Figure~\ref{fig:BaekLee_OddCases}.

We can solve for the up-left slide value $s_n$ to use in the $n^{th}$ row based on equating the length of the left diagonal with the triangle base as follows:

\begin{eqnarray*}
n - 1 + (1 + s_n) = n - (n - 1)s_n + 1
\end{eqnarray*}

\noindent From this, we find that $s_n = \frac{1}{n}$, and therefore the method can cover a $T_{n+d}$ with $n^2+3$ homothetic unit triangles if $d \le \frac{1}{n}$. The larger up-left slide value obtained with 3 extra unit triangles (instead of 2) shortens the width of base covered by the first $n$ up-triangles, but the shortage is accommodated by the width of the $T_2$ (because the last two up-triangles don't overlap).

As shown in Theorem~\ref{thm:BaekLee_Odd}, this method is optimal for the $n^2+3$ case using a variant of the proof for the $n^2+2$ case.

\begin{theorem}[Baek and Lee~\cite{BaekLee}]\label{thm:BaekLee_Odd}
If $d > \frac{1}{n}$, then a triangle $T_{n+d}$ having a selected side length of $n+d$ cannot be covered by $n^2 + 3$ homothetic unit triangles.
\end{theorem}
\begin{proof}
Referring to the notations and definitions in and supporting Theorem~\ref{thm:BaekLee_Even}, let $N = n^2 + 3$ so that the $S_1, \ldots, S_N$ unit $H$-triangles assumed to cover a $T_{n+d}$ has one more unit $H$-triangle than in the $n^2+2$ case. The reference function $r$ is unchanged, i.e., it is still assigned to be $r = f_{_{T_{n}}} + 1$ (because we must have $r \in \mathcal{T}$ so we cannot add $\frac{1}{2}$ to $r$ because $\frac{1}{2} \notin \mathcal{T}$). Therefore, since $r$ still represents only $n^2 + 2$ unit $H$-triangles but $f_{_{S}}$ now represented $n^2 + 3$ unit $H$-triangles, we expect that $\int h = \int f_{_{S}} - \int r = \frac{1}{2}$. However, from the proof of Theorem~\ref{thm:BaekLee_Even}, we still have:

\begin{eqnarray*}
h(t) \ge g(t) \ge nd - 1 + max(0, d-t)
\end{eqnarray*}

\noindent For the $n^2+3$ case, is advantageous to replace the term $max(0, d-t)$ with 0 because that is the maximum it contributes to the righthand side of the inequality. Then, with the assignments $g^{\prime} = 0$ and $\delta = nd - 1$, the logical steps of Corollary 16, Lemma 14 and Lemma 15 in~\cite{BaekLee} give the expectation that $\int h \ge 1$, contradicting the preceding expectation that $\int h = \frac{1}{2}$. Thus, it is infeasible to cover a $T_{n+d}$ with only $n^2 + 3$ homothetic unit triangles when $d > \frac{1}{n}$.
\end{proof}

\section{New Triangle Covering Method (Odd Cases)}\label{sec:NewTriangleCoveringMethod_Odd}

In Theorem~\ref{thm:Generalized_Odd}, we generalize the logic for the $n^2+3$ case that appeared in Theorem~\ref{thm:BaekLee_Odd} to establish upper bounds for all odd cases.

\begin{theorem}\label{thm:Generalized_Odd}
If $d > \frac{p}{n}$ (with $1 \le p < n \in \mathbb{N}$ and $d \in (0, 1)$), then a triangle $T_{n+d}$ having a selected side length of $n+d$ cannot be covered by $n^2 + 2p + 1$ homothetic unit triangles.
\end{theorem}
\begin{proof}
Referring to the notations and definitions in and supporting Theorem~\ref{thm:BaekLee_Odd}, let $N = n^2 + 2p + 1$ so that the $S_1, \ldots, S_N$ unit $H$-triangles assumed to cover a $T_{n+d}$ has $n^2 + 2p + 1$ homothetic unit triangles. Let $r = f_{_{T_{n}}} + p$ as we did in Theorem~\ref{thm:CoverEqTriangle_Optimal_Even}. This assignment to $r$ results in the assignment $\delta = nd - p$. Provided that $\delta > 0$, the desired contradiction is reached between the expectations of $\int h \ge 1$ and $\int h = \frac{1}{2}$, just as in Theorem~\ref{thm:BaekLee_Odd}. Thus, it is infeasible to cover a $T_{n+d}$ (of any shape, via affine transformations) using only $n^2 + 2p + 1$ homothetic unit triangles if $d > \frac{p}{n}$.
\end{proof}

To finish answering the open question from~\cite{BaekLee}, we must now show that the upper bounds in Theorem~\ref{thm:Generalized_Odd} are tight for each value of $p$ by presenting a covering method that can meet them.

One possible approach to covering a $T_{n+d}$ with $n^2 + 2p + 1$ homothetic unit triangles when $p > 1$ is to add a pair of unit triangles to each trapezoidal row above the $n^{th}$ row, plus the 3 unit triangles on the $n^{th}$ row as in Figure~\ref{fig:BaekLee_OddCases}. It is an interesting exercise to derive the threshold for $d$ that this method achieves. By using the same slide values as in \textsc{CoverTriangle} (Full Even) for each row $j$ before $n$, one can derive that the $n^{th}$ row receives the slightly wider slide value of $s_n = \frac{j}{n^2}$. For example, Figure~\ref{fig:BaekLee_OddCases} is drawn with $n = 3$, so the slide value depicted is $s_3 = \frac{1}{3}$, but if row $j=2$ were to also receive 2 extra unit triangles, then $s_2 = \frac{1}{3}$ and $s_3 = \frac{2}{9}$. Given the slide values on and before the last row, one can derive the threshold for $d$ to be:

\begin{eqnarray*}
d \le 1 - \left(\frac{j(n-1)}{n^2}\right)
\end{eqnarray*}

\noindent Observing that $p = n - j + 1$ and simplifying, we find that this method can cover a $T_{n+d}$ with $n^2 + 2p + 1$ homothetic unit triangles when:

\begin{eqnarray*}
d \le \frac{p}{n} - \frac{p-1}{n^2}
\end{eqnarray*}

\noindent Due to the subtrahend, we are not able to achieve a contradiction using the logic of Theorem~\ref{thm:Generalized_Odd} because the value of $\delta$ is not greater than 0 except when $p=1$ (i.e., except for the $n^2+3$ case). This does not necessarily mean that the method is non-optimal, but the upper bounds from Theorem~\ref{thm:Generalized_Odd} do \emph{suggest} that there may be a better method.

A better covering method can indeed be achieved by using the 3 extra unit triangles in the first row $j > 1$ that receives any extra unit triangles, and then by using 2 extra unit triangles on trapezoidal rows $j+1$ to $n$. It is a more interesting exercise to derive the threshold for $d$ for this alternative method. An inductive proof of the efficacy of this covering method can be achieved based on using a slide value of $s_j = \frac{1}{j}$ for the $j^{th}$ row (the one with 3 extra unit triangles), and $s_k = \frac{j^2 - 1}{j^2} \times \frac{j}{k(k+1)}$ for each row $j < k \le n$. Noting again that $p = n - j + 1$, one can derive that the upper bound for $d$ in terms of $p$ for this method is:

\begin{eqnarray*}
d \le \frac{p}{n+1} + \frac{1}{(n+1)(n-p+1)}
\end{eqnarray*}

\noindent This upper bound for $d$ is also not large enough to support reaching a contradiction because once again $\delta$ is not positive when $p > 1$. However, it is an improved upper bound, which suggests that we have moved in the proper direction by moving upward with the row in which the 3 extra unit triangles are added.

\subsection{The Main Idea (Odd Cases)}\label{subsec:MainIdeaMethod_Odd}

Continuing to move upward with the row in which to add 3 extra unit triangles presents a challenge because the only higher row is row 1, the top triangular row. If the 3 extra unit triangles are applied to row 1, then the calculated slide value is 1, which means that the top triangular row becomes a $T_2$. This is problematic because then each lower row then requires 2 extra unit triangles and \emph{no sliding} just to be wide enough to fit with the row above it. Taking this out to the $n^{th}$ row, we get a total of $n^2 + 2n + 1$ unit triangles. In other words, we get the na\"{\i}ve covering method. The challenge, then, is to make trapezoidal rows wider \emph{without} adding 2 extra unit triangles to them. To solve this challenge, we exploit a technique appearing in the second Conway-Soifer method in~\cite{ConwaySoiferGeom, ConwaySoifer}, which uses a \textbf{\emph{down-right slide operation}} to make trapezoidal rows that are shorter but wider. In this operation, each down-triangle slides downward and rightward by the slide value, which also pushes all subsequent unit triangles rightward. The main idea of the new covering method for odd cases consists of the following elements:

\begin{enumerate}
  \item Let row 1 be a $T_2$ covering (using 3 extra unit triangles, i.e., 1 pair plus 1 more).
  \item For each row $k$ from $2$ to $n-1$, add an extra down-triangle and up-triangle.
  \item For row $n$, do not add an extra pair of unit triangles, but perform the down-right slide operation with a slide value of $\frac{1}{n}$.
\end{enumerate}

An example of the above method in operation appears in Figure~\ref{fig:NewTriangleCoverMethod_OddCases}. The top row, which is row 1, has a side length of 2 because it receives 3 extra unit triangles that are marked +1, +2 and +3. The trapezoidal rows starting at row 2 receive 2 extra unit triangles, marked +1 and +2. With no sliding, the extra 2 unit triangles make each of these trapezoidal rows wide enough to fit with the row above it. To simplify the main idea, we have only applied the down-right slide operation to the last row.

\begin{figure}[!ht]
\centering

\includegraphics[width=4in,height=3.3in]{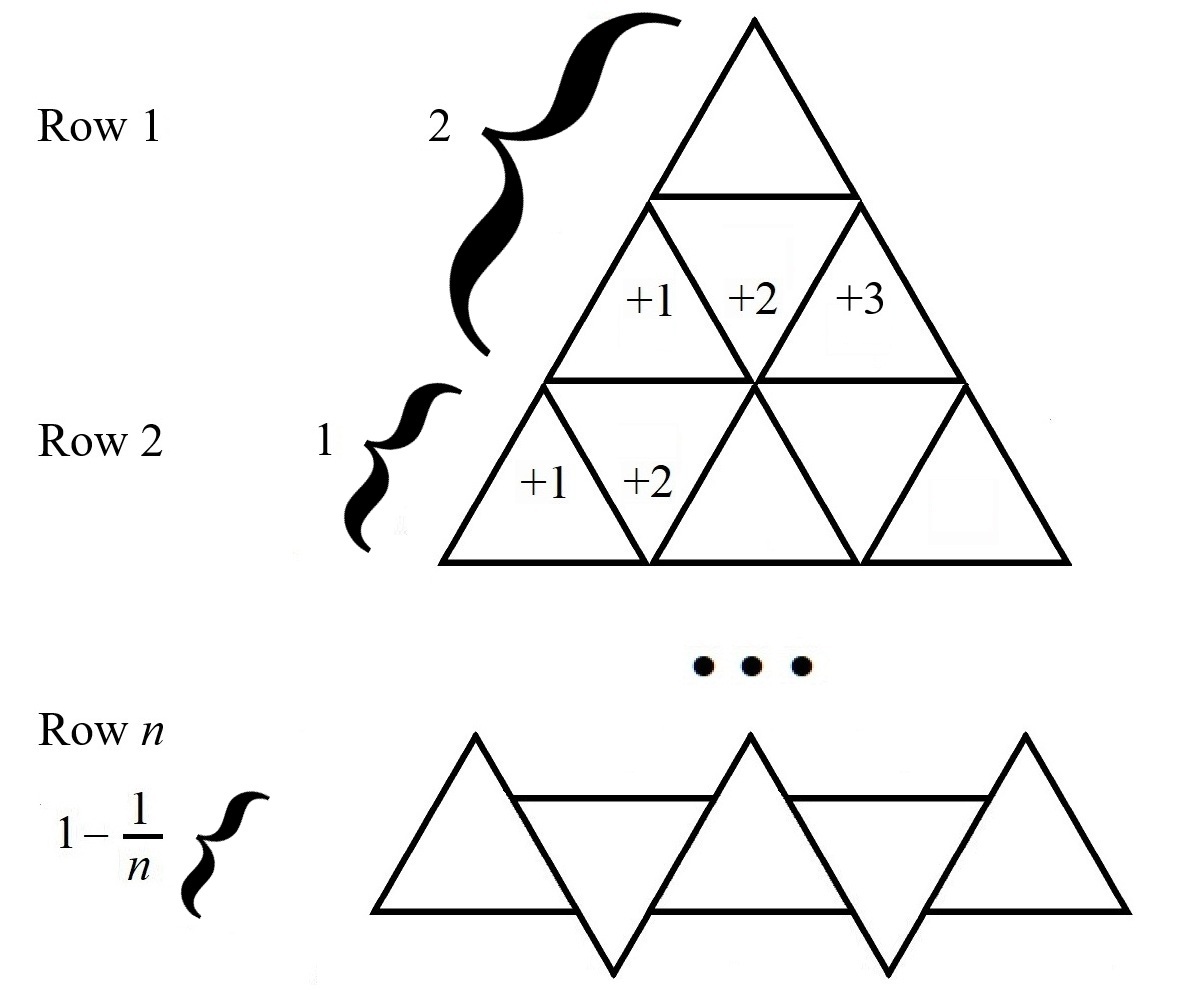}

\caption{Depiction of the main idea for covering a large triangle $T_{n+d}$ with an odd number of extra unit triangles. In this depiction, $n=3$, $d = \frac{2}{3}$, and $n^2+5$ unit triangles are used so that only the last row is made wider by the down-right slide operation rather than by adding an extra pair of unit triangles to it.
\label{fig:NewTriangleCoverMethod_OddCases}}
\end{figure}

When the down-right slide operation has been applied (see row $n$ in Figure~\ref{fig:NewTriangleCoverMethod_OddCases}), \textbf{\emph{down-teeth}} are the parts of down-triangles that jut below the baseline containing the bases of the up-triangles. Similarly, \textbf{\emph{up-teeth}} are the parts of up-triangles that jut above the line containing the horizontal legs of all down-triangles in the row.

The key insight is that the $n^{th}$ trapezoidal row only needs to cover a trapezoid with diagonal length $d$ and base length $n+d$ because the first $n-1$ rows of the covering have the size of a $T_n$. The base of the $n^{th}$ trapezoidal row already includes 1 unit of width for each of $n$ up-triangle bases, and it also includes the horizontal widths of $n-1$ down-teeth. Hence, the extra width $d$ can be spread across the $n-1$ down-teeth, yielding a slide value of $\frac{d}{n-1}$. Based on this slide value, the $n^{th}$ trapezoidal row can fully cover a trapezoid with diagonal length $1 - \frac{d}{n-1}$. Since the diagonal length must also be $d$, we solve $d =  1 - \frac{d}{n-1}$ to obtain $d = 1 - \frac{1}{n}$ as the largest value for $d$ that can be accommodated when using the down-right slide operation only on the $n^{th}$ row. For example, in Figure~\ref{fig:NewTriangleCoverMethod_OddCases}, the last row represents $n=3$, so the down-right slide value depicted is $\frac{1}{3}$. The diagonal length is $d = 1 - \frac{1}{3}$, and there are 2 down-triangles whose down-teeth contribute a total of $d = \frac{2}{3}$ to the base width of the $n^{th}$ row.

To better understand the upcoming generalization of this method, note that the width of the top of the $n^{th}$ row exactly matches the width of the bottom of the $T_n$ formed by the first $n - 1$ rows of the covering. Specifically, the width of the top of the $n^{th}$ row consists of the unit length horizontal legs of $n-1$ down-triangles, plus the horizontal width of $n$ up-teeth, each of width $\frac{1}{n}$. We can also see that this approach is optimal for the odd case $n^2 + 2n - 1$ based on Theorem~\ref{thm:Generalized_Odd} because we have added $p = n - 1$ pairs of unit triangles, one pair per row except the $n^{th}$ row. Hence, we have

\begin{eqnarray*}
d = 1 - \frac{1}{n} = \frac{n-1}{n} = \frac{p}{n}
\end{eqnarray*}

\subsection{Using Fewer Pairs of Unit Triangles (Odd Cases)}\label{subsec:UsingFewerUnitTriangles_Odd}

To generalize the method of the prior subsection to all odd cases, we choose the row $j$ on which to start performing the down-right slide operation in lieu of adding 2 unit triangles, and we define the magnitude of the slide value used on each row that does not receive 2 extra unit triangles. Specifically, the term \textbf{\textsc{CoverTriangle} (Basic Odd)} refers to the triangle covering method having the following elements:

\begin{enumerate}
  \item Let row 1 be a $T_2$ covering (using 3 extra unit triangles, i.e., 1 pair plus 1 more).
  \item For each row $k$ from $2$ to $j-1$, add an extra down-triangle and up-triangle.
  \item For each row $k$ from $j$ to $n$, do not add an extra pair of unit triangles, and instead perform the down-right slide operation with a slide value of $\frac{j-1}{k(k-1)}$
\end{enumerate}

\begin{theorem}\label{thm:CoverEqTriangle_Odd_Basic}
Given a triangle $T_{n+d}$ having a selected side length of $n+d$, with $n \in \mathbb{N}$ and $d \in (0, 1)$, and a row $1 < j \le n \in \mathbb{N}$, \textsc{CoverTriangle} (Basic Odd) covers the $T_{n+d}$ if $d \le 1 - \sum_{k=j}^{n} s_k$.
\end{theorem}
\begin{proof}
All rows up to and including $j-1$ form a $T_j$, which has a base length of $j$. For the first row in which the down-right slide operation is used, i.e. for $k=j$, the top of row $k$ also has width $j$ because it contains the horizontal legs of $j-1$ down-triangles plus $j$ up-teeth each of width $\frac{j-1}{k(k-1)}$, which have a total width of 1 since $k=j$. The widths of the bottom of each row $k \ge j$ and the top of row $k+1$ are also equal. Specifically, the bottom of row $k$ has a width consisting of $k$ up-triangle bases and $k-1$ down-teeth widths, as follows:

\begin{eqnarray*}
k + (k-1) \times \frac{j-1}{k(k-1)} = k + \frac{j-1}{k}
\end{eqnarray*}

\noindent By comparison, the top of row $i = k+1$ has a width consisting of the horizontal leg widths of $i-1$ down-triangles plus $i$ up-teeth widths, as follows:

\begin{eqnarray*}
i-1 + i \times \frac{j-1}{i(i-1)} = k + \frac{j-1}{k}
\end{eqnarray*}

\noindent Hence, \textsc{CoverTriangle} (Basic odd) fully covers any equilateral triangle $T_{n+d}$ provided $d$ is not too large. We must have $j \ge 2$ because the first 3 extra unit triangles are added on row 1, so row 2 is the earliest row on which to start performing down-right slide operations. The value of $d$ must not exceed the value of 1 minus the sum of the down-right slide values used on rows $j$ through $n$ of the covering that the \textsc{CoverTriangle} (Basic Odd) method generates. Therefore, we have the constraint $d \le 1 - \sum_{k=j}^{n} s_k$. The theorem follows for all triangles via affine transformations.
\end{proof}

\subsection{Selecting the Best \texorpdfstring{$j$}{\emph{j}} (Odd Cases)}\label{subsec:SelectingJ_Odd}

It is possible to poorly select the row $j$ on which \textsc{CoverTriangle} (Basic Odd) transitions to using down-right slide operations. Given a particular $T_{n+d}$ to cover, it is both necessary and possible to use the values of $n$ and $d$ to calculate the best row $j$ to select that minimizes the odd number of additional unit triangles that are added by \textsc{CoverTriangle} (Basic Odd). Specifically, Theorem~\ref{thm:CoverEqTriangle_SelectingJ_Odd} characterizes the minimum value of $j$ that can be selected to ensure that \textsc{CoverTriangle} (Basic Odd) covers a given $T_{n+d}$.

\begin{theorem}\label{thm:CoverEqTriangle_SelectingJ_Odd}
A triangle $T_{n+d}$ having a selected side length of $n+d$, with $n \in \mathbb{N}$ and $d \in (0, 1)$, can be covered by $n^2 + 2(j-1) + 1$ homothetic unit triangles with $j = \lceil dn + 1 \rceil$ and not with $j = \lceil dn + 1 \rceil - 1$.
\end{theorem}
\begin{proof}
For equilateral triangles, we begin with the inequality constraint for $d$ from Theorem~\ref{thm:CoverEqTriangle_Odd_Basic} and reduce to the simplest terms, as follows:

\begin{eqnarray*}
d \le 1 - \sum_{i=j}^{n} s_i &=& 1 - \left( \sum_{i=j}^{n} \frac{j-1}{i(i-1)} \right)\\
  &=& 1 - \left(j-1\right) \times \sum_{i=j}^{n} \left( \frac{1}{i-1} - \frac{1}{i} \right)\\
  &=& 1 - \left(j-1\right) \times  \left( \frac{1}{j-1} - \frac{1}{n} \right) \\
  &=& \frac{j-1}{n}
\end{eqnarray*}

\noindent Solving for $j$, we find that we must select a row such that:

\begin{eqnarray*}
j \ge dn + 1
\end{eqnarray*}

Since $j$ must be an integer, we use the ceiling function to ensure that down-right slide operations are performed on few enough rows that the sum of the down-right slide values does not exceed the value of $1 - d$. Therefore, we select:

\begin{eqnarray*}
j = \lceil dn + 1 \rceil
\end{eqnarray*}

Next, we observe that the value $\lceil dn + 1 \rceil - 1$ cannot be selected because it is less than $dn + 1$. Furthermore, for the vanishingly few triangles that $j$ is assigned the value $n+1$, the $T_{n+d}$ can be covered with $n^2 + 2n + 1$ homothetic unit triangles because that is the number that are in a $T_{n+1}$ covering produced by reverting to the na\"{\i}ve covering method. The theorem follows for all triangles via affine transformations.
\end{proof}

To facilitate exposition, the term \textbf{\textsc{CoverTriangle} (Full Odd)} refers to the version of \textsc{CoverTriangle} (Basic Odd) that selects $j$ using $j = \lceil dn + 1 \rceil$.

\section{Proof of Optimality for the Odd Cases}\label{sec:OptimalOddCases} In Theorem~\ref{thm:CoverEqTriangle_Optimal_Odd}, we now prove that the upper bounds obtained in Theorem~\ref{thm:Generalized_Odd} are tight because \textsc{CoverTriangle} (Full Odd) meets them.

\begin{theorem}\label{thm:CoverEqTriangle_Optimal_Odd}
\textsc{CoverTriangle} (Full Odd) is an optimal method for covering a $T_{n+d}$ with $n^2 + 2p + 1$ homothetic unit triangles for each $1 \le p < n \in \mathbb{N}$.
\end{theorem}
\begin{proof}
From the proof of Theorem~\ref{thm:CoverEqTriangle_SelectingJ_Odd}, \textsc{CoverTriangle} (Full Odd) selects an integer $j$ with the constraint $j \ge dn + 1$. The row $j$ on which down-right slide operations start is also the first row in which pairs of unit triangles are not added. Therefore, $p = j - 1$. Substituting into the constraint and simplifying, we have $d \le \frac{p}{n}$. Per Theorem~\ref{thm:Generalized_Odd}, if $d > \frac{p}{n}$, it is infeasible to cover a $T_{n+d}$ (of any shape, via affine transformations) with only $n^2 + 2p + 1$ homothetic unit triangles for each $1 \le p < n \in \mathbb{N}$.
\end{proof}

\section{Optimal Consolidated Method}\label{sec:ConsolidatedMethod} The term \textbf{\textsc{CoverTriangle} (Full)} refers to a method that covers a $T_{n+d}$ by selecting and using the one of \textsc{CoverTriangle} (Full Even) and \textsc{CoverTriangle} (Full Odd) that will use the fewest extra homothetic unit triangles. The decision is based on the following calculations:


\begin{eqnarray*}
j_{odd} &=& \lceil dn + 1 \rceil\\
p_{odd} &=& j - 1\\
j_{even} &=& \lfloor (1 - d)(n + 1) \rfloor\\
p_{even} &=& n - j + 1
\end{eqnarray*}


\noindent If $p_{odd} < p_{even}$, then \textsc{CoverTriangle} (Full) selects and uses \textsc{CoverTriangle} (Full Odd) and otherwise \textsc{CoverTriangle} (Full Even) is selected and used.

\begin{theorem}\label{thm:CoverEqTriangle_Optimal_Full}
\textsc{CoverTriangle} (Full) is an optimal method for covering a $T_{n+d}$ with homothetic unit triangles.
\end{theorem}
\begin{proof}
By Theorems~\ref{thm:CoverEqTriangle_Optimal_Even} and~\ref{thm:CoverEqTriangle_Optimal_Odd}, $p_{even}$ and $p_{odd}$ are the least number of, respectively, extra pairs and one plus extra pairs of homothetic unit triangles more than $n^2$ that can cover a given $T_{n+d}$. When $p_{odd} < p_{even}$, then $2p_{odd} + 1 < 2p_{even}$. Otherwise, $p_{odd}$ indicates the use of at least as many pairs as $p_{even}$, so $2p_{odd}+1 > 2p_{even}$.
\end{proof}

\section{Conclusion and Future Work}\label{sec:ConclusionFutureWork}

In this paper, we have presented a new method for covering a given triangle $T_{n+d}$, with a selected side length of $n \in \mathbb{N}$ plus $d \in (0, 1)$, using $n^2+k$ homothethic unit triangles. By this new method, we proved that we have answered in the affirmative an open question from~\cite{BoyerAlwaysCoverTriangles} about whether there was a better triangle covering method than the one in that paper. Also, in~\cite{BaekLee}, Baek and Lee developed a proof technique for analyzing triangle coverings, and they used it to establish tight upper bounds for the cases $n^2+1$, $n^2+2$, and $n^2+3$. They then opened the question of finding tight upper bounds for $n^2+k$ for $k > 3$. We answered this open question by extending their proof technique to the higher cases and then showing that the upper bounds were tight based on the new method meeting them. The analysis of how to arrange the triangles to cover a $T_{n+d}$ was based on multiples of the reciprocals of the \emph{pronic numbers}, including for one case the fractional reciprocals of the \emph{triangular numbers}.

Since the problem of optimally covering large triangles with homothetic unit triangles has been fully solved in this paper, future work should focus on making an improvement by lifting the restriction of using all homothetic unit triangles. For example, we know that an improvement is possible if the restriction is lifted even for just one unit triangle. Namely, Baek and Lee used their proof technique to prove that there are no $T_{n+d}$ that can be covered with only $n^2+1$ homothetic unit triangles~\cite{BaekLee}, yet in~\cite{BoyerNonequilateralTriangles} a method is shown that can cover some $T_{n+d}$ using $n^2$ homothetic unit triangles plus 1 unit triangle that is not homothetic to but is similar to the $T_{n+d}$ and is even arranged with one leg parallel to the $x$-axis (a requirement in the Baek and Lee proof technique). Furthermore, the method in~\cite{BoyerNonequilateralTriangles} was shown to work for all non-equilateral triangles, but if the requirement of similarity of the unit triangles is replaced by the requirement that all unit triangles have the same area, then there are even some large equilateral triangles than can be covered with $n^2$ homothetic unit equilateral triangles plus one non-equilateral triangle having the same area as a unit equilateral triangle. A new proof technique, or a further extension of the Baek-Lee proof technique, is required to account for the effect of using one or more non-homothetic unit triangles to help cover a large triangle.

\clearpage

\bibliographystyle{plainurl}
\bibliography{OptimallyCoveringLargeTriangles.bib}

\end{document}